\documentclass[conference]{IEEEtran}
\IEEEoverridecommandlockouts
\usepackage{cite}
\usepackage{amsmath,amssymb,amsfonts}
\usepackage{graphicx}
\usepackage{textcomp}
\usepackage{amsfonts}
\usepackage{amssymb}
\usepackage{amsthm}
\usepackage{bbm}
\usepackage{pstool}
\usepackage{algorithm}
\usepackage[noend]{algpseudocode}
\usepackage[margin=10pt,labelfont=bf]{caption}
\usepackage{tikz}
\usetikzlibrary{shadows.blur}
\usetikzlibrary{arrows,positioning} 
\usetikzlibrary{chains,shapes.multipart}
\usetikzlibrary{shapes,calc}
\usetikzlibrary{automata,positioning}

\newtheorem{proposition}{Proposition}
\newtheorem{definition}{Definition}
\usepackage{xcolor}
\def\BibTeX{{\rm B\kern-.05em{\sc i\kern-.025em b}\kern-.08em
    T\kern-.1667em\lower.7ex\hbox{E}\kern-.125emX}}
\begin{document}

 \title{Risk-Sensitive Reinforcement Learning for URLLC Traffic in Wireless Networks}

\author{Nesrine Ben-Khalifa$^\diamond$, Mohamad Assaad$^\diamond$, M\'erouane Debbah$^*$\\
\IEEEauthorblockA{ $^\diamond$TCL chair on 5G, Laboratoire des Signaux et Syst\`emes (L2S), Centrale Sup\'elec, France\\
$^*$Huawei Technologies, Boulogne-Billancourt, France}}


\maketitle
\begin{abstract}
In this paper, we study the problem of dynamic channel allocation for URLLC traffic in a multi-user multi-channel wireless network where urgent packets have to be successfully transmitted in a timely manner. We formulate the problem as a \textit{finite-horizon} Markov Decision Process with a stochastic constraint related to the QoS requirement, defined as the {\textit{packet loss rate}} for each user. We propose a novel weighted formulation that takes into account both the total expected reward (number of successfully transmitted packets) and the {\textit{risk}} which we define as the QoS requirement violation. First, we use the value iteration algorithm to find the optimal policy, which assumes a perfect knowledge of all the model parameters, namely the channel statistics. We then propose a Q-learning algorithm where the controller learns the optimal policy without having knowledge of neither the CSI nor the channel statistics. We illustrate the performance of our algorithms with numerical studies.
\end{abstract}

\vspace{0.2cm}

\section{Introduction}

In the fifth generation (5G) wireless networks, there are new service categories with heterogeneous and challenging requirements, among them the Ultra Reliable Low Latency (URLLC) traffic \cite{urllc_bennis}, designed for delay and reliability sensitive applications like real-time remote control, autonomous driving, and mission-critical traffic. In URLLC traffic, the End-to-End (E2E) latency defined by 3GPP is lower than $1$ ms along with a reliability requirement of $1-10^{-5}$ to $1-10^{-9}$ \cite{urllc_bennis,urllc_2018}. 


A plausible solution to address the latency requirement issue is to make transmissions without Channel State Information (CSI) knowledge at the transmitter side. To increase reliability, exploiting frequency diversity is beneficial, and occurs by making parallel transmissions of the same packet over different subcarriers in an Orthogonal Frequency Division Multiplexing (OFDM) system where each subcarrier experiences different channel characteristics.

However, this solution is costly in terms of system capacity. Therefore, the number of parallel transmissions should not be fixed in advance but should rather be variable and depending on many parameters such as the position of a user in the cell, or the statistics about his packet losses over the previous time slots. For example, if a user experienced a high number of packet losses in the previous time slots, it should be allocated a high number of subchannels to increase his success probability, whereas a user with a low number of dropped packets may be assigned a low number of subcarriers. Hence, it is crucial to design efficient dynamic schemes able to adapt the number of parallel transmissions for each user to his experienced QoS.



In this work, we study the problem of dynamic channel allocation for URLLC traffic in a multi-user multi-channel wireless network under QoS constraints. A channel here refers to a frequency band or a subcarrier in an OFDM system, and the QoS is related to the {\textit{packet loss rate}} for each user, defined as the average number of dropped packets. Besides, we introduce the notion of {\textit{risk}} related to the violation of the QoS requirements; more precisely, a risk occurs or equivalently, a risk state is reached when the QoS requirement is violated for a user. Furthermore, we consider that the transmitter does not have neither the CSI nor the channel statistics at the transmission moment. In fact, due to the urgency of URLLC packets mentioned previously, there is not enough time for the BS to make channel estimation and probing techniques like in conventional wireless communications.

\subsection{Related Work}

The issue of deadline-constrained traffic scheduling has been investigated by several works including \cite{deadline,spawc,wi_opt2018,wiopt2017}. For example, in \cite{wiopt2017}, the authors study the problem of dynamic channel allocation in a single user multi-channel system with service costs and deadline-constrained traffic.  They propose online algorithms to enable the controller to learn the optimal policy based on Thompson sampling for multi-armed bandit problems.
The MDP framework and reinforcement learning approaches for downlink packet scheduling are considered in \cite{deadline,wi_opt2018,assaad_2009,assaad_2,assaad_3,assaad_4}. In \cite{deadline}, the authors propose an MDP for deadline-constrained packet scheduling problem and use dynamic programming to find the optimal scheduling policies. The authors do not consider QoS constraints in the scheduling problem.

Most risk-sensitive approaches consist in analyzing higher order statistics than the average metric such as the variance of the reward \cite{urllc_bennis,risk1,risk2,risk_assaad}. For instance, a risk-sensitive reinforcement learning is studied in \cite{risk_mmwave} in millimeter-wave communications to optimize both the bandwidth and transmit power. The authors consider a utility (data rate) that incorporates both the average and the variance to capture the tail distribution of the rate, useful for the reliability requirement of URLLC traffic. The authors do not exploit frequency diversity.

In this work, we consider an alternative approach to the risk which consists in minimizing the risk state visitation probability. In fact, due to the stochastic nature of the problem (time-varying channel and random arrival traffic in our context),  giving a low reward to an undesirable or a risk-state may be insufficient to minimize the probability of  visiting such state \cite{risk_sensitive}. Therefore, in addition to the maximization of the total expected reward, we propose to consider a second criterion which consists in minimizing the probability of visiting risk states where a risk state here is related to the violation of QoS requirements. 
\subsection{Addressed Issues and Contribution}
In this work, we address the following issues:
\begin{itemize}
\item[$\bullet$] We formulate the dynamic channel allocation problem for URLLC traffic as a \textit{finite-horizon} MDP wherein the state represents the QoS of the users, that is, the average number of dropped packets or {\textit{packet loss rate}} of the users. The decision variable is the number of channels to assign to each user. We define a \textit{risk state} as any state where the QoS requirement is violated for at least one user. Besides, we define a  stochastic constraint related to the risk state visitation probability.
\item[$\bullet$] Assuming the channel statistics are known to the controller, we use the finite-horizon value iteration algorithm to find the optimal policy to the weighted formulation of the problem, which takes into account both the total expected reward over the planning horizon and the {\textit{risk}} criterion (QoS requirement violation probability).  
\item[$\bullet$] When the channel statistics are unknown to the controller, we propose a reinforcement learning algorithm (Q-learning) for the weighted formulation of the problem, which enables the controller to learn the optimal policy. We illustrate the performance of our algorithms with numerical studies.

\end{itemize}
\subsection{Paper Structure}
In Section \ref{model}, we present the system model for the multi-user multi-channel wireless network with URLLC packets and time-varying channels along with the QoS definition. In Section \ref{CMDP}, we introduce the constrained MDP formulation with all its components. In Section \ref{algorithm}, we present both the finite-horizon value iteration algorithm and the reinforcement learning algorithm. Section \ref{numerical} is devoted to numerical results. Finally, we conclude the paper in Section \ref{conclusion}.
\section{System Model}\label{model}

We consider a multi-user multi-channel  wireless network where URLLC packets have to be transmitted over time-varying and fading channels. Due to the strict latency requirement of URLLC packets in 5G networks mentioned previously, there is not enough time for the BS to estimate the channel, and the packets are then immediately transmitted in the absence of CSI at the transmitter side. When a packet is successfully decoded, the receiver sends an acknowledgment feedback, which is assumed to be instantaneous and error-free. We consider a centralized controller which dynamically distributes the channels to the users based on their QoS (see Fig. \ref{fig:scheme}).
\begin{figure}[t!]
\centering
 \begin{tikzpicture}[ 
   node distance = 4mm and 16mm,
mynode/.style = {shape=signal, signal to=west and east,
                 draw, 
                 text width=1.6cm, align=flush center, 
                 inner xsep=0mm, inner ysep=3mm}
                 ]
                 
      \tikzstyle{shape2} = [rectangle, rounded corners, minimum width=1.3cm, minimum height=0.6cm,text centered, draw=black]
\node (B) [mynode]   {Controller};
\node (ue1) [shape2,above right=0.6cm and 2.2cm of B]   { user $1$ };
\node (uek) [shape2, below right=0.6cm and 2.2 cm of B]   { user $K$ };
\coordinate[above=0.5cm of ue1]     (a1);
\coordinate[above=0.5cm of uek]     (ak);
\coordinate[above=0.1cm of ue1]     (a12);
\coordinate[above=0.1cm of uek]     (ak2);
\coordinate[right=0.00cm of B.east]     (outc);
\coordinate[right=of ue1]        (outue1);
\coordinate[right=of uek]    (outuek);
\coordinate[left=0.8cm of ue1]    (inue1);
\coordinate[left=0.8cm of uek]    (inuek);
\coordinate[right=0.2cm of B] (outc2);
\coordinate[right=1cm of ue1]        (outue12);
\coordinate[right=1cm of uek]        (outuek2);
\coordinate[right=1.7cm of uek]        (outuek3);
\coordinate[below right=0.3cm and 1.7cm of uek]        (outuek4);
\coordinate[below left=0.3cm and 4.9cm of uek]        (outuek5);
\coordinate[below left=1cm and 0.5cm of ue1](vdot);
\coordinate[left=0.65cm of B.west] (outuek6);

\draw[dashed, black,thick,shorten >=1mm] (ue1.east) --  (outue12) -- +(0.7,0) -- (outuek3) ;
\draw[->,dashed, black, thick] (uek.east) --  (outuek2) -- (outuek3) --  (outuek4) -- (outuek5) -- (outuek6) -- (B.west) ;
\node[ above left= 0.1 cm and -0.001cm of outue1] {$\rho_1(t+1)$};
\node[ above left= 0.1 cm and -0.001cm of outuek] {$\rho_K(t+1)$};
\node[above right=0.1cm and 1.8cm of vdot]{\large{$\vdots$}};
\node[above=-0.02cm of a1]{$a_1(t)$};
\node[above=-0.02cm of ak]{$a_K(t)$};
\draw[->,black,thick] (a1) -- (a12) ;
\draw[->,black,thick] (ak) -- (ak2) ;
\draw[->,black, thick] (outc) --  (outc2) -- (inue1) --  (ue1);
\draw[->,black,  thick] (outc2) --   (inuek) -- (uek);
\node[above right=0.1cm and -0.1cm of inue1] {$\ell_1(t)$};
\node[above right=0.1cm and -0.1cm of inuek] {$\ell_K(t)$};
\end{tikzpicture}
\caption{Dynamic allocation of channels $(\ell_1,..,\ell_K)$ to the users based on their QoS $(\rho_1,..,\rho_K)$.}\label{fig:scheme}\end{figure}
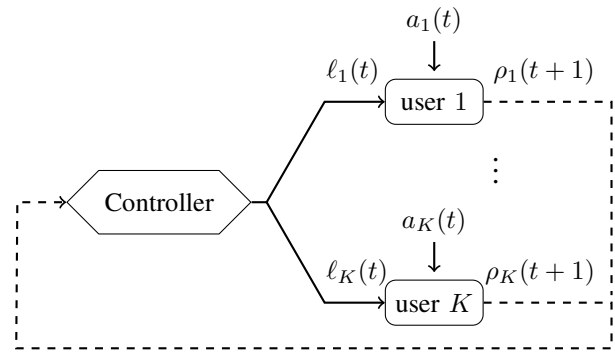

Furthermore, we make the following assumptions:

{\textit{Packet arrival process}}: the packet arrival process is considered as an independent and identically distributed (i.i.d.) random process over a finite set $\mathcal{I}=\{ 0,1,..,A_{max} \}$, where $A_{max}$ is a positive constant, and is identical for all the users. Let $\alpha_a$ denote the probability that $a \in \mathcal{I}$ packets arrive for a given user at the beginning of a time slot.

{\textit{Deadline-constrained traffic:}} regarding the strict URLLC latency requirement specified by 3GPP (lower than 1 ms), each packet has a lifetime of one time slot and can either be served or dropped; if there are available channels, the packet will be transmitted, otherwise, it will be dropped because after one time slot it becomes outdated and useless. Furthermore, one packet is transmitted per channel.

{\textit{Channel model}}: we consider  i.i.d. Bernoulli channels with a mean $\mu \in [0,1]$. In millimeter-wave communications, the links are characterized by their intermittence and high sensitivity, and  this channel model reflects the existence of a light-of-sight (LOS) channel state\cite{wiopt2017,mmwave}. To increase reliability, a user can be assigned more channels than the number of waiting packets (depending on his experienced QoS). Some packets are then simultaneously sent over multiple parallel channels.
 

\textit{Channel split:} for each user, all the packets are equally important: when the number of available channels is larger than that of waiting packets, we assume that some packets are picked uniformly at random to be replicated. A packet is obviously more likely to be successfully transmitted when sent over many channels simultaneously. However, assigning more channels to a user will affect the QoS experienced by the other users. Note that channel split across the packets (which occurs in the same manner for all the users) should not be confused with the channel split across the users (which takes into account the QoS perceived by the users). 

For user $k$, the distribution of available channels $\ell_k$ over the waiting packets $a_k$ occurs as follows:
each packet is transmitted over ($\ell_k \wedge a_k$) channels and may be furthermore replicated once with a probability ($\frac{\ell_k \vee a_k}{a_k}$), where the symbol $\ell_k\wedge a_k$ denotes the larger integer $m$ such that $m a_k \leqslant \ell_k$, and $\ell_k\vee a_k$ denotes the remaining integer of the division of $\ell_k$ by $a_k$.

The probability that a packet is successfully transmitted given that there are $a_k$ waiting packets at the transmitter and $\ell_k$ assigned channels can then be expressed by
\begin{eqnarray}
\nu_k(a_k,\ell_k)&=& \left (1- \frac{\ell_k \vee a_k }{a_k} \right ) \left (1-(1-\mu)^{(\ell_k \wedge a_k)} \right ) \nonumber \\ && + \left ( \frac{\ell_k \vee  a_k}{a_k} \right ) \left (1- (1-\mu)^{ 1+ (\ell_k \wedge a_k)} \right).
\end{eqnarray}
The expected number of successfully transmitted packets for user $k$ is then given by
\begin{eqnarray}
\mathbb{E} \left [ N_k(\ell_k) \right ]=\sum_{a_k \in \mathcal{I}}^{}{a_k \alpha_{a_k} \nu_k(a_k,\ell_k)}.
\end{eqnarray}
{\textit{ QoS criterion:}} for each user $k$, we define the packet loss rate at time slot $t$, $\rho_k(t)$, as follows
\begin{eqnarray}\label{ro}
\rho_k(t)=\frac{1}{t}\sum_{i=0}^{t-1}{\frac{n_k(i)}{a_k(i)}}, \quad t\geqslant 1,
\end{eqnarray}
where $n_k(t)$ denotes the number of lost packets for user $k$ at time slot $t$.
Note that $\rho_k \in [0,1]$ ($n_k(t) \leqslant a_k(t)$). A packet is lost when either of the two following events occurs:
\begin{itemize}
\item[(i)] it is not transmitted because of insufficient available channels,
\item[(ii)] is transmitted but ACK feedback is not received. 
\end{itemize}
The parameter $\rho_k$ reflects the QoS perceived by user $k$: higher values of $\rho_k$ mean a higher number of lost packets and poor QoS whereas lower values of $\rho_k$ mean good QoS. To ensure good QoS for the users, the resource allocation scheme should take account of their experienced QoS and keep this parameter values for all users within an acceptable range.

Finally, the decision variable is the number of channels associated to each user $k$ at each time slot, denoted by $\ell_k$, which satisfies
\begin{eqnarray}\label{contrainte}
\sum_{k=1}^{K}{\ell_k(t)}=L,
\end{eqnarray}
where $L$ denotes the number of available channels.
\section{Constrained MDP Framework}\label{CMDP}
The stochastic nature of the wireless channel incites us to consider an MDP framework to solve the decision problem. In this section, we first introduce the constrained MDP formulation along with its components. We then derive the optimality equations.
\subsection{Model Formulation}
We define the following \textit{finite-horizon} MDP 
\begin{itemize}
\item[$\bullet$] \textit{State Space}: is the finite set $\mathcal{T}\times\mathcal{S}$ where $\mathcal{T}=\{0,..,T \}$, $\mathcal{S}=\left \{\rho_1 \times .. \times \rho_K \right \}$, $\rho_k$ for $k=1,..,K$ is defined in \eqref{ro}, and the symbol $\times$ stands for the Cartesian product.
\item[$\bullet$]  \textit{Action Space}: is the finite set $\mathcal{L}=\{ (\ell_1,..,\ell_K) \mbox{ satisfying } \eqref{contrainte} \}$, where $\ell_k$ denotes the number of channels assigned to user $k$. 
\item[$\bullet$] \textit{Reward}: we define the reward $r$ at time slot $t$, when the controller chooses action $\ell \in \mathcal{L}$ in state $s_t$, as the expected total number of successfully transmitted packets over all the users, that is,
\begin{eqnarray}\label{reward}
r(s_t,\ell)=\mathbb{E}\left [ {\sum_{k=1}^{K}{N_k}(\ell_k)} \right ].
\end{eqnarray}
Note that the reward depends only on the number of channels allocated for each user (the action), and not on the current state $s_t$. Besides, the reward is a non-linear function of the action.
\item[$\bullet$]\textit{Transition Probabilities:} 

First, we define the probability that $n$ packets are lost for user $k$ as a function of the number of waiting packets $a_k$ and the number of assigned channels $\ell_k$ at a given time slot as follows
\begin{eqnarray*}
\sigma_k(n,a_k,\ell_k)=\binom{a_k}{n} \Big ( 1-\nu_k(a_k,\ell_k)   \Big )^n \nu_k(a_k,\ell_k)^{a_k-n},
\end{eqnarray*}
where $n \leqslant a_k$ and $\binom{a_k}{n}$ denotes the binomial coefficient. 
 The state transition probability for user $k$ is given by
\begin{eqnarray}
p(\rho'_k \,|\, \rho_k(t), \ell_k)=\alpha_{a_k} \sigma_k(n,a_k,\ell_k),
\end{eqnarray}
where
\begin{eqnarray}
\rho'_k=\frac{t}{t+1}\rho_k+\frac{1}{t+1}\frac{n}{a_k}.
\end{eqnarray}
Finally, let $s_{t+1}=\rho'_1\times..\times \rho'_K$ and $s_t=\rho_1\times..\times \rho_K$, the transition probability from state $s_t$ to state $s_{t+1}$ given when action $l$ is taken, is then given by
\begin{eqnarray}
p(s_{t+1} \, | \, s_t,\ell)=\prod_{k=1}^{K}{p(\rho'_k \, | \, \rho_k(t),\ell_k)}.
\end{eqnarray}
\end{itemize}
Regarding the strict requirements of URLLC packets described earlier, we introduce in the following the notion of a \textit{risk-state}.
\begin{definition}
We define a \textit{risk state} any state  where $\rho_k > \rho_{max}$ for any $k \in \{1,..,K\}$ with $\rho_{max}>0$ is constant fixed by the controller. The set of risk states $\Phi$ is then,
\begin{eqnarray*}
\Phi=\{\rho_1\times .. \times \rho_K  \mbox{ where there $\exists$ } k \mbox{ such that } \rho_k > \rho_{max} \}.
\end{eqnarray*}
\end{definition}
Besides, a risk-state is an absorbing state, that is, the process ends when it reaches a risk state \cite{risk_sensitive}.

A deterministic policy $\pi$ assigns at each time step and for each state an action. Our goal is to find an optimal deterministic policy $\pi^*$ which maximizes the  total expected reward $\mathcal{V}_T^{\pi}(s)$ given by
\begin{eqnarray}\label{v_function}
\mathcal{V}_T^\pi(s)=\mathbb{E}^{\pi}\left [ \sum_{t=0}^{T}{ r(s_t,\pi(s_t)) | \, s_0=s } \right ],
\end{eqnarray}
with the reward $r$ is defined in \eqref{reward}, while satisfying the QoS constraint given by
\begin{eqnarray}\label{constraint}
\eta^{\pi}(s)<w,
\end{eqnarray}
where $\eta^{\pi}(s)$ denotes the probability of visiting a risk state over the planning horizon, given that the initial state (at time slot $0$) is $s$ and policy $\pi$ is followed, and $w$ is a positive constant. Formally,
\begin{eqnarray}
\eta^{\pi}(s)=\mathbb{P}^{\pi}\; ( \exists \,t \mbox{ such that } s_{t} \in \Phi |s_0=s).
\end{eqnarray}
In order to explicitly characterize $\eta^{\pi}(s)$, we introduce in the following the risk signal $\overline{r}$.
\begin{definition}
We define a risk signal $\overline {r}$ as follows
\begin{equation}
\overline {r}(s_t,\ell_t,s_{t+1})=
\left \{
\begin{array}{cc}
1& \quad \mbox { if } s_{t+1} \in \Phi \\
0&  \quad \mbox{ otherwise,}
\end{array}
\right.
\end{equation}
where $s_t$ and $\ell_t$ denote the state and action at time slot $t$, respectively, and $s_{t+1}$ denotes the subsequent state.
\end{definition}
\begin{proposition}\label{prop1}
\textbf{The probability of visiting a risk-state, $\eta^{\pi}(s)$,  is given by} 
\begin{eqnarray}\label{eq:prob}
\eta^{\pi}(s)=\overline{\mathcal V}_T^{\pi}(s),
\end{eqnarray}
\textbf{where we set}
\begin{eqnarray}\label{v_var}
\overline{\mathcal V}_T^{\pi}(s)=\mathbb{E}^{\pi}\left [ \displaystyle \sum_{t=0}^{T}{\overline{r}\left (s_t, \pi(s_t), s_{t+1} \right ) \, | \, s_0=s} \right ].
\end{eqnarray}
\end{proposition}
\begin{proof}
The random sequence $\overline r(t=0)$, $\overline r(t=1)$,.., $\overline r(t=T)$ may contain $1$ if a risk state is visited, otherwise all its components are equal to zero (recall that a risk state is an absorbing state). Therefore, $\sum_{t=0}^{T}{\overline{r}(t)}$ is a Bernoulli random variable with a mean equal to the probability of reaching a risk state, that is, relation \eqref{eq:prob} holds.
\end{proof}
\subsection{Optimality Equations}
By virtue of Proposition  \ref{prop1}, we associate a state value function $\overline{\mathcal V}_T^{\pi}$ to the probability of visiting a risk state. Now, we define a new weighted value function $\mathcal{V}_{\xi, T}^{\pi}$, which incorporates both the reward and the risk, as follows
\begin{eqnarray}\label{weighted_v}
\mathcal{V}_{ \xi, T}^{\pi}(s)=\xi \mathcal{V}_T^{\pi}(s)-\overline{ \mathcal{V}}_T^{\pi}(s),
\end{eqnarray}
where $\xi>0$ is the weighting parameter, determined by the risk level the controller is willing to tolerate. The function $\mathcal{V}^{\pi}_{\xi,T}$ can be seen as a standard value function associated to the reward $\xi r -\overline r$. The case $\xi=0$ corresponds to a {\textit{minimum-risk}} policy whereas the case $\xi \rightarrow \infty$ corresponds to a {\textit{maximum-value}} policy.

Let $\Pi$ denote the set of deterministic policies, and define 
\begin{eqnarray*}
\mathcal{V}_T^*(s)=\underset{ \pi \in \Pi}{ \mbox{ max } } \mathcal{V}_T^{\pi}(s), \quad \mathcal{ \overline{V}}_T^*(s)=\underset{ \pi \in \Pi}{ \mbox{ min } } \mathcal{\overline {V}}_T^{\pi}(s), 
\end{eqnarray*}
\begin{eqnarray*}
\mathcal{V}^*_{\xi,T}(s)=\underset{ \pi \in \Pi}{ \mbox{ max } } \mathcal{V}^{\pi}_{\xi,T}(s).
\end{eqnarray*}
Besides, we define $u^{\pi}_t$, $\overline{u}_t^{\pi}$, and $u^{\pi}_{\xi,t}$ for $0\leqslant t \leqslant T $ respectively by
\begin{eqnarray}
u^{\pi}_t(s)&=&\mathbb{E}^{\pi}\left [  \sum_{i=t}^{T}{r(s_i,\pi(s_i)) } |\, s_t=s \right ], \label{u1} \\
\overline{u}^{\pi}_t(s)&=&\mathbb{E}^{\pi}\left [  \sum_{i=t}^{T}{\overline{ r }(s_i,\pi(s_i),s_{i+1}) } |\, s_t=s \right ], \label{u2} \\
u^{\pi}_{\xi,t}(s)&=&\xi u^{\pi}_t(s)-\overline{u}^{\pi}_t(s). \label{u3}
\end{eqnarray}
Note that $\mathcal{V}^{\pi}_{T}$ incorporates the total expected reward over the entire planning horizon whereas $u_t$ incorporates the rewards from decision epoch $t$ to the end of the planning horizon only. Besides, $\overline{u}_t(s)$ is the probability of visiting a risk state given that at time $t$ the system is in state $s\in \{ \mathcal{S} / \Phi \}$, and is thus a measure of the risk.

The optimality equations are given by (the proof is similar to that in \cite{puterman}, chap. 4 and skipped here for brevity)
\begin{eqnarray}
u_t^*(s)&=&\underset{\ell \in \mathcal{L}}{ \mbox{max }} \Big \{ r (s_t,\ell)+  \sum_{j \in \mathcal{S}}^{}{ p(j|s_t,\ell)u_{t+1}^*(j) } \Big \} \label{u1} \;\;\;\;\; \\
{\overline {u}}_t^*(s)&=&\underset{\ell \in \mathcal{L}}{ \mbox{min }} \Big \{ \sum_{j \in \mathcal{S}}^{}{ p(j|s_t,\ell) \Big ( \overline{r} (s_t,\ell,j)+\overline{u}_{t+1}^*(j)  \Big ) } \Big \} \;\; \; \;\; \; \;\;\label{u2} \\
u_{\xi,t}^*(s)&=&\underset{\ell \in \mathcal{L}}{ \mbox{max }} \Big \{ \sum_{j \in \mathcal{S}}^{}{ p(j|s_t,\ell)  \Big ( \xi r (s_t,\ell) - \overline{r} (s_t,\ell,j) }\nonumber \\
 && + u_{\xi,t+1}^*(j) \Big) \Big \}, \label{u3}
\end{eqnarray}
for $t=0,..,T-1$.  For the boundary conditions, that is at time slot $T$, $u^*_T(s)$, $\overline{u}^*_T(s)$, and $u_{\xi,T}^*(s)$ are set to zero for each $s\in\mathcal{S}$.

In a non-risk state, the reward $r$ is given in \eqref{reward} and the risk signal is equal to zero whereas in a risk state the reward $r$ is set to zero and the risk signal $\overline r$ is set to one. 
\section{Algorithm Design}\label{algorithm}
In this section, we present two algorithms: (i) finite-horizon value iteration algorithm which assumes that all the model parameters are known to the controller, namely the channel statistics (model-based algorithm), and (ii) reinforcement learning algorithm which does not require the controller knowledge of channel statistics (model-free algorithm). 
\subsection{Value Iteration Algorithm}
In order to find a policy that maximizes the weighted value function defined in \eqref{weighted_v}, we use the value iteration algorithm \cite{puterman}. In this algorithm, we proceed backwards: we start by determining the optimal action at time slot $T$ for each state, and successively consider the previous stages, until reaching time slot $0$ (see Algorithm 1).
%
\begin{algorithm}\label{alg:value}
\caption{Finite-Horizon Value Iteration Algorithm}\label{v_function}
\begin{algorithmic}[1]
\State Initialization: \textbf{for each s} 
\State $u_T^*(s) \gets 0$, $\overline{u}^*_T(s) \gets  0$, $u_{\xi,T}^*(s) \gets 0$
\State \textbf{Endfor}
\State $t\gets T-1$
\State  \textbf{while} $t\geqslant 0$ 
\State  \textbf{For each s }
\State \textbf{update} $u_t^*(s)$, $\overline {u}_t^*(s)$, and ${u}_{\xi, t}^*(s)$ according to \eqref{u1}, \eqref{u2}, and \eqref{u3}, respectively
\State \textbf{EndFor}
\State  $t \gets t-1$
\State  \textbf{EndWhile}
\end{algorithmic}
 \end{algorithm}
\subsection{Risk-Sensitive Reinforcement Learning Algorithm}
\begin{figure}[t]
\centering
\tikzset{
    >=stealth',
    punkt/.style={
           rectangle,
           rounded corners,
           draw=black,
           text width=6.5em,
           minimum height=2em,
           text centered},
    pil/.style={
           ->,
           shorten <=2pt,
           shorten >=2pt,},
           pil2/.style={
           ->,
           shorten <=2pt,
           shorten >=2pt,}         
}
\begin{tikzpicture}
\tikzset{
myarrow/.style={->, >=latex', shorten >=2pt},
mylabel/.style={text width=9em, text centered} 
} 
\node [minimum height=1.9cm, minimum width=2.5cm,  top color=white, bottom color=orange!25,
 rounded corners, text centered, draw=black] (node1)  { \parbox{1.95cm}{ Environment \\ (Wireless Channel)} };
\node[ below= 1.56cm of node1, minimum height=1.7cm, minimum width=3cm,top color=white, bottom color=yellow!20,
rounded corners,  text centered, draw=black]  (node2) {Learning Controller};
\node[left= -0.3 cm of  node2] (g) {}
edge[pil, black, bend left=50] (node1.west);
\node[right= -0.3 cm of  node1] (g) {}
edge[pil, black, bend left=50] (node2.east);
\draw [->, red] (node1.east)+(0,0.2) to   [out=1,in=0]  (1.6,-3.75) ;
\draw [->, dashed] (node1.east)+(0,-0.8) to   [out=80,in=50]  (1.6,-3.01) ;
\draw (1.2,-1.6) node[text=black] {new};%
\draw (1.2,-2.0) node[text=black] {state};%
\draw (3.01,-1.6) node[text=red] {risk};
\draw (3.3,-2.09) node[text=red] {signal $\overline{r}$};
\draw (-2.7,-1.65) node[] {action};
\draw (-2.5,-3.5) node[text=black] {Interaction};
\draw (3.5,-3.5) node[text=black] {Observation};
\draw (2.35,-1.8) node[text=black] {\large{$r$}};
\end{tikzpicture}\caption{Reinforcement Learning Model.}\label{fig:learning}
\end{figure}
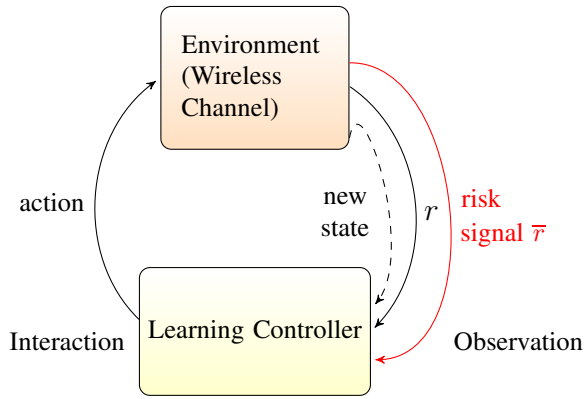
During the learning phase, the controller gets estimates of the value of each state-action pair. It updates its estimates through the interaction with the environment  where at each iteration it performs an action and then observes the reward, risk signal $\overline {r}$, and the next state (see Fig. \ref{fig:learning}). 

The learning controller chooses an action at each learning step following the $\varepsilon$-greedy policy, that is, it selects an action that maximizes its current estimate with probability $1-\varepsilon$, or a random action with probability $\varepsilon$. The parameter $\varepsilon$ captures the exploration-and-exploitation trade-off: when $\varepsilon \rightarrow 0$, the controller tends to choose an action that maximizes its  current state's estimated value; whereas when $\varepsilon \rightarrow 1$, the controller tends to choose randomly an action and to favor the exploration for optimality.

The state-action value function is given by \cite{Watkins1992,RL}
\begin{eqnarray*}
Q^{\pi}(s_t,\ell)=r(s_t,\ell)+\sum_{j \in \mathcal{S}}{p(j|s_t, \ell) \, {u}^{\pi}_{t+1}(j) },
\end{eqnarray*}
where the first term denotes the immediate reward, that is the number of successfully transmitted packets over all the users,  when the action $l$ is performed in state $s_t$; and the second term denotes the expected reward when the policy $\pi$ is followed in the subsequent decision stages.
Similarly to the state-action value function associated to the reward, we define the state-action value function associated to the risk $\overline{Q}^{\pi}$ as
\begin{eqnarray*}
\overline{Q}^{\pi}(s_t,\ell)=\sum_{j \in \mathcal{S} }{p(j|s_t,\ell) \left (  \overline  {r}(s_t,\ell,j) + \overline {u}^{\pi}_{t+1}(j) \right )}.
\end{eqnarray*}
Note that the introduction of the signal risk $\overline{r}$ enabled us to define a state-action value function, $\overline{Q}$ to the risk.

Besides, the state-action value function associated to the weighted formulation, $Q_{\xi}^{\pi}$, is given by
\begin{eqnarray*}
Q_{\xi}^{\pi}(s_t,\ell)=\xi Q^{\pi}(s_t,\ell)-\overline{Q}^{\pi}(s_t,\ell).
\end{eqnarray*}
Finally, the Q-function updates at the learning step $n$ (which should not be confused with the decision epoch $t$) are given by \cite{Watkins1992}
\begin{eqnarray}
Q^{(n+1)}(s_t,\ell)& \leftarrow &\big [ 1-\alpha_n (s_t,\ell) \big ] Q^{(n)}(s_t,\ell)   + \nonumber  \\ &&\alpha_n (s_t,\ell) \big [ r  + \underset{\ell \in \mathcal{L} }{ \mbox{ max } }\{  Q^{(n)}(s_{t+1},\ell) \} \big ],\;\;\;\; \;\; \label{q1} 
\end{eqnarray}
\begin{eqnarray}
\overline{Q}^{(n+1)}(s_t,\ell)& \leftarrow & \big [1-\alpha_n (s_t,\ell) \big] \overline{Q}^{(n)}(s_t,\ell)   + \nonumber \\  &&  \alpha_n (s_t,\ell) \big [ \overline r +\underset{\ell \in \mathcal{L} }{ \mbox{ min } }\{ \overline{Q}^{(n)}(s_{t+1},\ell) \} \big ], \;\;\;\; \; \;  \label{q2} 
\end{eqnarray}
and,
\begin{eqnarray}
Q_{\xi}^{(n+1)}(s_t,\ell)&  \leftarrow &\big[ 1-\alpha_n (s_t,\ell) \big] Q_{\xi}^{(n)}(s_t,\ell)+ \nonumber \\ && \alpha_n (s_t,\ell) \big [ \xi r-\overline r   +\underset{\ell \in \mathcal{L}}{\mbox{ max }} \{ Q_{\xi}^{(n)}(s_{t+1},\ell) \} \big ], \nonumber \\ \label{q3}
\end{eqnarray}
where $\alpha_n(s_t,\ell)$ denotes the learning rate parameter at step $n$ when the state $s_t$ and action $\ell$ are visited. 

The learning algorithm converges to the optimal state-action value function when each state-action pair is performed infinitely often and when the learning rate parameter satisfies for each $(s_t,\ell)$ pair (the proof is given in \cite{Watkins1992,RL2} and skipped here for brevity),
\begin{eqnarray*}
\sum_{n=1}^{\infty}{\alpha_n}(s_t,\ell)= \infty, \quad \mbox{ and } \quad \sum_{n=1}^{\infty}{\alpha_n^2(s_t,\ell)} < \infty.
\end{eqnarray*}
In this case, the Q-functions are related to the value functions as follows
 \begin{eqnarray*}
 \underset{\ell \in \mathcal{L}}{\mbox{max}} \left \{ Q(s_t,\ell) \right \}=u_t^{*}(s_t),  \quad 
 \underset{\ell \in \mathcal{L}}{\mbox{min}} \left \{ \overline {Q}(s_t,\ell) \right \}=\overline{u}_t^*(s_t),  
 \end{eqnarray*}
 \begin{eqnarray*}
 \underset{\ell \in \mathcal{L}}{\mbox{max}} \left \{ Q_{\xi}(s_t,\ell) \right \}=u_{\xi,t}^*(s_t). 
 \end{eqnarray*}
 
When a risk state is reached during the learning phase, the system is restarted according to the uniform distribution to a non-risk state. In addition, when $t \geqslant T$, we consider that an artificial absorbing state is reached and we reinitialize $t$ (see Algorithm 2). 
\begin{algorithm}\label{alg:value}
\caption{Q-learning Algorithm}\label{q_learning}
\begin{algorithmic}[1]
\State Initialization $t \gets 0$, $s_0\gets s$, $ n \gets 1$, 
\State \textbf{ for each $\ell \in \mathcal{L}$}
\State $Q(s_0,\ell)\gets 0$, $\quad$ $\overline{Q}(s_0,\ell) \gets 0$, $\quad$ $Q_{\xi}(s_0,\ell) \gets 0$
\State \textbf{End for}
\State \textbf{Repeat}
\State observe current state $s_t$
\State select and perform action $\ell$ in state $s_t$
\State observe the new state $s_{t+1}$, reward  $r$ and the risk $\overline r$
\State \textbf{update the Q-functions } $Q(s_t,l)$, $\overline{Q}(s_t,l)$, $Q_{\xi}(s_t, \ell)$ according to \eqref{q1}, \eqref{q2}, \eqref{q3} respectively
\State $t \gets t+1$
\State $n \gets n +1 $
\State \textbf{update} $\alpha_n$ 
\State \textbf{if} $t=T$, \textbf{then} $t\gets 0$ \quad  \emph{artificial absorbing state reached}
\State \textbf{if} $s_{t} \in \Phi$, \textbf{then} $s_t\sim\mbox{Unif} \{\mathcal{S}/ \Phi \} $ \;  \emph{absorbing state reached}
\State \textbf{until convergence}
\end{algorithmic}
 \end{algorithm}
\section{Performance Evaluation}\label{numerical}
In this section, we present the numerical results obtained with the value iteration and the learning algorithms in a variety of scenarios. We consider the setting of two users along with a number of channels $L=5$. For the arrival traffic, we consider the following truncated Poisson distribution
\begin{equation}
\mbox{Prob}(a=m)=
\left \{
\begin{array}{cc}
   \frac{\lambda^m/m!}{ \sum_{i=0}^{A_{max}}{\lambda^i/i!}} & \mbox{ if } m \leqslant A_{max} \\
\mbox{ zero } & \mbox{otherwise,}
\end{array}
\right.
\end{equation}
where $\lambda=3$ and $A_{max}=6$. The mean of the Bernoulli channel $\mu$ and the value of the parameter $\rho_{max}$ throughout this section are fixed to $0.6$ and $0.55$ respectively.
\subsection{Minimum-risk vs maximum-value policy}
First, we compare the performance of the minimum-risk policy (obtained when $\xi=0$), maximum-value policy (obtained when $\xi \rightarrow \infty$), weighted policy (when $\xi>0$), and the fixed policy which consists is assigning the same number of channels for each user at each time slot ($\ell_1=2$ and $\ell_2=3$).
\begin{figure}[t!]
\psfrag{a}{\large{$u_t(s)$}}
\psfrag{b}{\large{$\overline{u}_t(s)$}}
\psfrag{time}{time}
\psfrag{rrr}{$ \overline{\pi}^*$}
\psfrag{vvv}{$\pi^*$}
\psfrag{fff}{$\pi_f$}
\psfrag{www}{$\pi_{\xi}^*$}
\psfrag{0}{$0$}
\psfrag{1}{$1$}
\psfrag{2}{$2$}
\psfrag{3}{$3$}
\psfrag{4}{$4$}
\psfrag{5}{$5$}
\psfrag{6}{$6$}
\psfrag{7}{$7$}
\psfrag{8}{$8$}
\psfrag{9}{$9$}
\psfrag{1.2}{\small{$1.2$}}
\psfrag{1.4}{\small{$1.4$}}
\psfrag{1.6}{\small{$1.6$}}
\psfrag{1.8}{\small{$1.8$}}
\psfrag{2.2}{\small{$2.2$}}
\psfrag{2.4}{\small{$2.4$}}
\psfrag{2.6}{\small{$2.6$}}
\psfrag{2.8}{\small{$2.8$}}
\psfrag{0.1}{\hskip -0.2 cm { $0.1$}}
\psfrag{0.2}{\hskip -0.2 cm { $0.2$}}
\psfrag{0.3}{\hskip -0.2 cm { $0.3$}}
\psfrag{0.4}{\hskip -0.2 cm { $0.4$}}
\psfrag{0.5}{\hskip -0.2 cm { $0.5$}}
\psfrag{0.6}{\hskip -0.15 cm { \small{$0.6$}}}
\psfrag{0.7}{\hskip -0.2 cm { $0.7$}}
\psfrag{0.8}{\hskip -0.2 cm { $0.8$}}
\psfrag{0.9}{\hskip -0.2 cm { $0.9$}}
\centering\includegraphics[width=9.7cm]{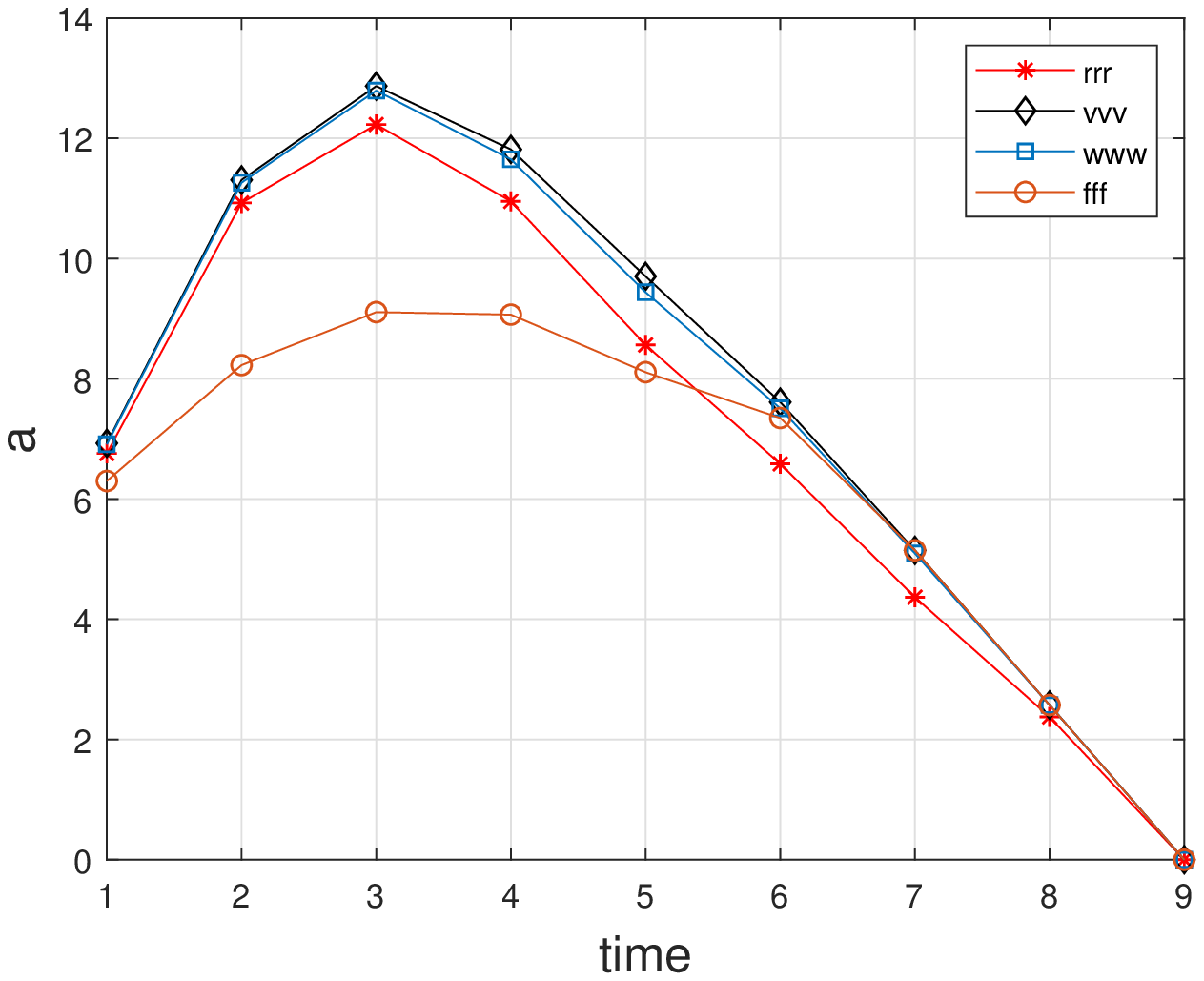} \centering\includegraphics[width=9.7cm]{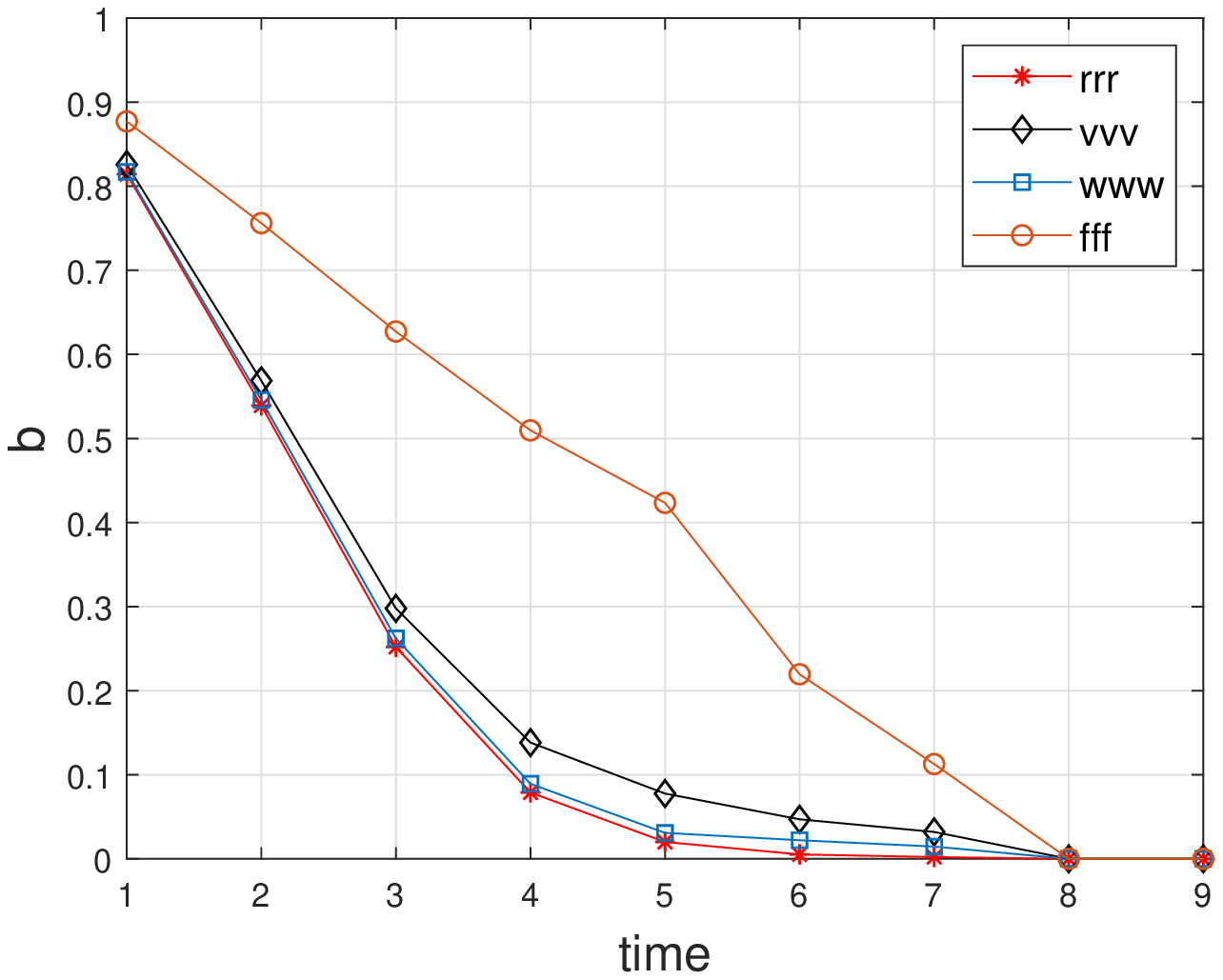}
\caption{Performance of the minimum-risk policy $\overline{\pi}^*$, the maximum-value policy $\pi^*$, the weighted-policy $\pi_{\xi}^*$ with $\xi=0.1$, and the fixed policy $\pi_f$. On the \textit{top}, $u_t(s)$, on the \textit{bottom}, $\overline{u}_t(s)$ where $s=0.3\times0$ and $T=9$.}\label{fig:risk}
\end{figure}

We depict in Fig. \ref{fig:risk}-\textit{top} the reward $u_t(s)$ given in \eqref{u1} as a function of time when $s=0.3\times0$ and different policies are followed. We observe that the maximum-value policy clearly outperforms the fixed and the minimum-risk policy. In Fig. \ref{fig:risk}-\textit{bottom} showing $\overline{u}_t(s)$ given in \eqref{u2}, we observe that the probability of visiting a risk-state when the fixed policy is followed is much higher than that obtained when the minimum-risk policy $\overline{\pi}^*$ is performed. For example, at the time step $t=5$, $\overline{u}_t(s)$ is equal to $0.42$ when the policy $\pi_f$ is performed whereas this value reduces to $0.02$ when the policy $\overline{\pi}^*$ is followed. In fact, the fixed policy does not take account of the experienced QoS of the users, and therefore, it is the policy which results in the highest risk-state visitation probability.  Besides, this probability decreases over time for all the policies. In fact, as time goes on, the probability of entering a risk-state over the remaining time steps decreases.

The reward $u_t(s)$ increases for the lower values of $t$ until reaching a maximum value and then it decreases, for all the policies. In fact, for the lower values of $t$, the probability of visiting a risk-state is high, and this affects the expected value of the reward (recall that in the risk state, the reward is equal to zero). As time goes on, this probability decreases, and thus the expected reward increases. However, at the further time steps, the number of remaining decision stages is low and hence the expected reward (total number of successfully transmitted packets over the remaining time slots) decreases.
\subsection{Learning}
In the learning algorithm, we simulate the wireless channel with a Bernoulli random variable with a number of trials equal to the number of channels associated to each packet for each user.
For the learning rate parameter $\alpha_n$, we considered the following expression \cite{q_rate}:
\begin{eqnarray}
\alpha_n=\frac{1}{ (1+n(s_t,\ell))^{\gamma}},
\end{eqnarray}
where $n(s_t,\ell)$ denotes the number of times the state-action pair $(s_t,\ell)$ was visited until iteration $n$, and $\gamma$ is a positive parameter $\in [0.5,1]$\cite{q_rate}.
\begin{figure}[t!]
\centering
\psfrag{ddddddddddd}{$\rho_1=0$}
\psfrag{eeeeeeeeeee}{$\rho_1=0.2$}
\psfrag{ggggggggggg}{$\rho_1=0.3$}
\psfrag{l}{$\ell_1^*$}
\psfrag{time}{time}
\psfrag{0}{$0$}
\psfrag{0.5}{}
\psfrag{1}{$1$}
\psfrag{1.5}{}
\psfrag{2}{$2$}
\psfrag{2.5}{}
\psfrag{3}{$3$}
\psfrag{3.5}{}
\psfrag{4}{$4$}
\psfrag{4.5}{}
\psfrag{5}{$5$}
\psfrag{5.5}{}
\includegraphics[width=9.2cm]{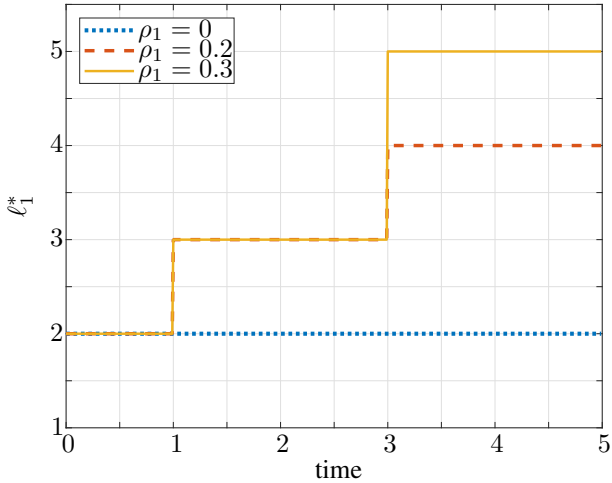}
\caption{Optimal policy $\ell_1^*$ as a function of time steps and $\rho_1$ with $\rho_2=0$ and $T=5$.}\label{optimal_policy}
\end{figure}

We depict in Fig.  \ref{optimal_policy} the optimal (minimum-risk) policy (number of channels to assign to user $1$ , $\ell_1 \in [0,..,5]$)  computed by the learning algorithm, as a function of time steps (decision epochs) and $\rho_1$, when $\rho_2$ is fixed to $0$. The figure shows a monotony property: the number of channels to assign to user $1$ increases with time and with $\rho_1$. In fact, as the QoS of user $1$ degrades ($\rho_1$ increases), more channels are assigned to it to compensate for this degradation; and as time goes on, this policy is more sensitive to this degradation as more channels are assigned for the same values of $\rho_1$, but at further time steps. 
\section{Conclusion}\label{conclusion}
In this work, we studied the problem of dynamic channel allocation for URLLC traffic in a multi-user multi-channel wireless network within a novel framework. Due to the stochastic nature of the problem related to time-varying, fading channels and random arrival traffic, we considered a finite-horizon MDP framework. We determined explicitly the probability of visiting a risk state and we wrote it as a cumulative return (risk signal). We then introduced a weighted global value function which incorporates two criteria: reward and risk. By virtue of the value iteration algorithm, we determined the optimal policy. Furthermore, we used a Q-learning algorithm to enable the controller to learn the optimal policy in the absence of channel statistics. We illustrated the performance of our algorithms with numerical studies, and we showed that by adapting the number of parallel transmissions in a smart way, the performance of the system can be substantially enhanced.

In the future work, we would like to take account of spatial diversity in the dynamic allocation scheme where both the BS and the user terminals can be equipped with multiple antennas to enhance the system performance.

\begin{thebibliography}{00}



\bibitem{assaad_2009}
R.~Aggarwal, M.~Assaad, C.~E. Koksal, and P.~Schniter.
{OFDMA} downlink resource allocation via {ARQ} feedback.
 In {\em Conference Record of the Forty-Third Asilomar Conference on
  Signals, Systems and Computers}, pages 1493--1497, Pacific Grove, CA, Nov
  2009.

\bibitem{assaad_4}
A.~Ahmad and M.~Assaad.
\newblock Joint resource optimization and relay selection in cooperative
  cellular networks with imperfect channel knowledge.
\newblock In {\em Proc. of IEEE 11th International Workshop on Signal
  Processing Advances in Wireless Communications (SPAWC)}.

\bibitem{assaad_3}
A.~Ahmad and M.~Assaad.
\newblock Optimal resource allocation framework for downlink {OFDMA} system
  with channel estimation error.
\newblock In {\em Proc. of IEEE Wireless Communication and Networking
  Conference}, pages 1--5, Sydney, April 2010.

\bibitem{mmwave}
M.~R. Akdeniz, Y.~Liu, M.~K. Samimi, S.~Sun, S.~Rangan, T.~S. Rappaport, and
  E.~Erkip.
\newblock Millimeter wave channel modeling and cellular capacity evaluation.
\newblock {\em IEEE Journal on Selected Areas in Communications},
  32(6):1164--1179, June 2014.

\bibitem{risk_assaad}
M.~Assaad, A.~Ahmad, and H.~Tembine.
\newblock Risk sensitive resource control approach for delay limited traffic in
  wireless networks.
\newblock In {\em Proc. of IEEE Global Telecommunications Conference
  (GLOBECOM)}, pages 1--5, Kathmandu, Dec 2011.

\bibitem{urllc_bennis}
Mehdi Bennis, M\'erouane Debbah, and H.~Vincent Poor.
\newblock Ultra-reliable and low-latency wireless communication: tail, risk and
  scale.
\newblock {\em Technical Report, ArXiv}, may 2018.

\bibitem{RL2}
D.~P. Bertsekas and J.~Tsitsiklis.
\newblock {\em Neuro-dynamics Programming}.
\newblock MIT Press, Cambridge, Massachusetts, 1996.

\bibitem{wiopt2017}
S.~Cayci and A.~Eryilmaz.
\newblock Learning for serving deadline-constrained traffic in multi-channel
  wireless networks.
\newblock In {\em Proc. of the 15th International Symposium on Modeling and
  Optimization in Mobile, Ad Hoc, and Wireless Networks (WiOpt)}, pages 1--8,
  May 2017.

\bibitem{wi_opt2018}
A.~Destounis, G.~S. Paschos, J.~Arnau, and M.~Kountouris.
\newblock Scheduling {URLLC} users with reliable latency guarantees.
\newblock In {\em Proc. of the 16th International Symposium on Modeling and
  Optimization in Mobile, Ad Hoc, and Wireless Networks (WiOpt)}, pages 1--8,
  May 2018.

\bibitem{deadline}
A.~Dua and N.~Bambos.
\newblock Downlink wireless packet scheduling with deadlines.
\newblock {\em IEEE Transactions on Mobile Computing}, 6(12):1410--1425, Dec
  2007.

\bibitem{q_rate}
Eyal Even-Dar and Yishay Mansour.
\newblock Learning rates for {Q}-learning.
\newblock {\em J. Mach. Learn. Res.}, 5:1--25, 2004.

\bibitem{risk_sensitive}
Peter Geibel and Fritz Wysotzki.
\newblock Risk-sensitive reinforcement learning applied to control under
  constraints.
\newblock {\em J. Artif. Int. Res.}, 24(1):81--108, jul 2005.

\bibitem{risk1}
A.~Gosavi.
\newblock Finite horizon {Markov} control with one-step variance penalties.
\newblock In {\em Proc. of the 48th Annual Allerton Conference on
  Communication, Control, and Computing}, pages 1355--1359, Sept 2010.

\bibitem{assaad_2}
N.~U. Hassan and M.~Assaad.
\newblock Dynamic resource allocation in multi-service {OFDMA} systems with
  dynamic queue control.
\newblock {\em IEEE Transactions on Communications}, 59(6):1664--1674, June
  2011.

\bibitem{urllc_2018}
H.~Ji, S.~Park, J.~Yeo, Y.~Kim, J.~Lee, and B.~Shim.
\newblock Ultra-reliable and low-latency communications in {5G} downlink:
  Physical layer aspects.
\newblock {\em IEEE Wireless Communications}, 25(3):124--130, June 2018.

\bibitem{risk2}
A.~Kumar, V.~Kavitha, and N.~Hemachandra.
\newblock Finite horizon risk sensitive {MDP} and linear programming.
\newblock In {\em Proc. of the 54th IEEE Conference on Decision and Control
  (CDC)}, pages 7826--7831, Dec 2015.



\bibitem{spawc}
N.~Nomikos, N.~Pappas, T.~Charalambous, and Y.~Pignolet.
\newblock Deadline-constrained bursty traffic in random access wireless
  networks.
\newblock In {\em SPAWC}, May 2018.

\bibitem{puterman}
Martin~L. Puterman.
\newblock {\em Markov Decision Processes: Discrete Stochastic Dynamic
  Programming}.
\newblock John Wiley \& Sons, 2005.

\bibitem{RL}
Richard~S. Sutton and Andrew~G. Barto.
\newblock {\em Reinforcement Learning: An Introduction}.
\newblock MIT Press, Cambridge, Massachusetts, 1998.

\bibitem{risk_mmwave}
T.~K. Vu, M.~Bennis, M.~Debbah, M.~Latva-aho, and C.~S. Hong.
\newblock Ultra-reliable communication in {5G} {mmWave} networks: A
  risk-sensitive approach.
\newblock {\em IEEE Communications Letters}, 22(4):708--711, April 2018.

\bibitem{Watkins1992}
Christopher J. C.~H. Watkins and Peter Dayan.
\newblock Q-learning.
\newblock {\em Machine Learning}, 8(3):279--292, May 1992.

\end{thebibliography}


\end{document}